\documentclass{eptcs}
\providecommand{\event}{LFMTP 2011} 
\usepackage{amsmath}
\usepackage{amsfonts}
\usepackage{amssymb}

\usepackage{stmaryrd}
\usepackage{verbatim}
\usepackage[latin1]{inputenc}

\usepackage{breakurl}

\usepackage{tikz}
\usetikzlibrary{matrix,arrows}

\usepackage{graphicx}
\usepackage{proof}

\usepackage{color}

\usepackage{pdfpages}

\newcommand{\ttt}{\texttt}
\newcommand{\llb}{\llbracket}
\newcommand{\rrb}{\rrbracket}

\newcommand{\ov}{\overline}

\newcommand{\botimes}[2]{\smash{\bigotimes^{{#1}}_{{#2}}} \vphantom{\bigotimes^{{#1}}_{{#1}}}}
\newcommand{\buplus}[2]{\smash{\biguplus^{{#1}}_{{#2}}} \vphantom {\biguplus^{{#1}}}}

\newcommand\celf{{Celf}}
\newcommand\clf{{\textsc{clf}}}
\newcommand\brs{\textsc{brs}}

\newtheorem{definition}{Definition}
\newtheorem{example}{Example}
\newtheorem{lemma}{Lemma}
\newtheorem{proof}{Proof}
\newtheorem{theorem}{Theorem}

\author{Maxime Beauquier and Carstern Sch\"{u}rmann
\institute{IT University of Copenhagen}\\
\email{[maxime.beauquier | carsten]@itu.dk}}
\date{}

\title{A Bigraph Relational Model
\thanks {\footnotesize This work was
   in part supported by NABITT grant 2106-07-0019 of the Danish Strategic Research
   Council.}}

\def\titlerunning{A Bigraph Relational Model}
\def\authorrunning{Maxime Beauquier, Carsten Sch\"{u}rmann}

\newif\iftechreport

\techreportfalse

\begin{document}

\iftechreport
\else
{\maketitle}
\fi

\begin{abstract}
  In this paper, we present a model based on relations for
  \emph{bigraphical reactive systems}~\cite{MILNER09}.  Its defining
  characteristics are that validity and reaction relations are
  captured as traces in a multi-set rewriting system.  The relational
  model is derived from Milner's graphical definition and directly
  amenable to implementation.

\end{abstract}

\section*{Introduction}

Milner's bigraphical reactive systems~\cite{MILNER09}, or
\textsc{brs}s in short, are formulated in terms of category theory.
They encompass earlier models, such as
\textsc{ccs}~\cite{Milner06PureBigraphs}, the
\(\pi\)-calculus~\cite{JensenMilner03BigraphsTransitions}, and Petri
nets~\cite{Milner03BigraphsPetriNets}.  However, as with other
categorical models, it is not immediately clear how to implement a
logical framework that could check, for example, the well-formedness
or correctness of a {\brs}, or just execute reaction rules.  On the
other hand there are mature implementations of logical frameworks,
e.g.\ \celf~\cite{Celf08} that already provide many of the algorithms
that one would need for such an implementation.  In particular, \celf
~provides support for linearity and concurrency using a kind of
structural congruence that arises naturally from the definition of
equivalence in the type theory \clf~\cite{CLF02}.  In this paper we
show that the two are deeply connected.  In particular, we formulate a
\emph{bigraph relational model} for \textsc{brs}s and demonstrate how
to piggy-bag on \celf's implementation by reusing algorithms, such as
unification, type checking, type inference, logic programming, and
multi-set rewriting.

A \textsc{brs} consists of a bigraph and a set of reaction rules.  The
bigraph consists of a \emph{place graph}, that usually models the
hierarchical (physical) structure of the concurrent system to be
modeled, and the \emph{link graph} that establishes the communication
structure between the different places.  By the virtue of this
definition alone, a bigraph does not have any dynamic properties.  It
is best understood as a snapshot of a concurrent system at a
particular point of time.

What makes a bigraph reactive is the accompanying set of reaction
rules. A reaction rule can be thought of as a rewrite rule, except
that the left and the right hand side are graphs rather than terms.
Consequently, matching the left hand side of a reaction rule with a
subgraph of the current bigraph is conceptually and computationally
not as straightforward as for example first-order unification.




As an alternative, we relate bigraphical reactive systems to something
that we understand well: unification modulo structural congruence in
the setting of {\clf}~\cite{SchackNielsen10lfmtp}.  {\clf} is a type
theory that conservatively extends the $\lambda$-calculus and serves
to model truly concurrent systems.  {\clf} follows the standard
judgements-as-types encoding paradigm, which means that derivations of
the validity of a {\brs}, traces of the operational semantics, etc.\
are encoded as {\clf} objects of the corresponding {\clf} type.  The
extensions include, for example, type families that are indexed by
objects, a dependent type constructor $\Pi x:A.\,B(x)$ as a
generalisation of the usual function type constructor, linear type
constructors, for example $A \multimap B$, that capture the nature of
resource consumption, and also a monadic type constructor $\{ A \}$.
This type is inhabited by objects of type $A$, such that two objects
are considered (structurally) congruent if and only if they differ
only in the order in which sub-terms are evaluated.


The \celf~\cite{Celf08} system is an implementation of {\clf} type
theory that provides a concept of logic variables that are logically
well understood in terms of linear contextual modal type theory and a
rich set of algorithms, including a sophisticated unification
algorithm, type inference, type checking, and a logical programming
language that supports both backward chaining proof search and forward
chaining multi-set rewriting.

The main contribution of this paper is a model for \textsc{brs} that
we call \emph{bigraph relational model} that follows closely the
graphical interpretation of the categorical model proposed by
Milner~\cite{MILNER09}.  Nodes, roots, sites, edges etc.\ are modelled
by multi-sets and the bigraph formation rules as rewrite rules.  A specific
bigraph is valid if and only if the corresponding multi-set can be rewritten to the empty set.
We also give a direct and elegant encoding of reaction rules as
rewrite rules. 
 Furthermore, we prove that {\celf}'s operational
interpretation of these rules is \emph{adequate} in the sense that it
coincides with the intended meaning of those rules.  As a consequence,
{\celf}'s multi-set rewriting engine implements the reactive behaviour
of a {\brs}.

For illustrative purposes, we choose Milner's bigraph encoding of
\textsc{ccs}~\cite{Milner06PureBigraphs} as running example, we sketch
the corresponding bigraph relational model.

\iftechreport
\else{More details are
  available in the technical report~\cite{Beauquier10tr} and the
  source code is available from \ttt{\url{www.itu.dk/~beauquie/brs}}.}
\fi

The remainder of this paper is structured as follow: in
Section~\ref{section:bigraph} we reiterate the formal definition of
the bigraph structure. In Section~\ref{section:brs} we define reaction
rules and therefore \textsc{brs}. Then, we define the bigraph
relational model in Section~\ref{section:encoding_bigraph}, and show
that the structures properties of bigraphs are respected.  Furthermore
we define encoding and interpretation functions, and show that the
bigraph relational model is adequate.  We define the encoding of the
reaction rules and prove adequacy in
Section~\ref{section:encoding_brs}. Finally, we show the
implementation in {\celf} in Section~\ref{section:implementation} before we
conclude and assess results.

\section{Bigraphs}
\label{section:bigraph}

The definition used in this paper can be found in~\cite{MILNER09}.  A
bigraph consists of a set of nodes.  Each node is characterised by a
type, which we call \emph{control}.  Each \emph{control} is defined by
a number of \emph{ports}.  We write \(K\) for the set of
\emph{controls}, and \(arity:K \to \mathbb{N}\) a map from controls to
the number of ports.  Together they form what we call the
\emph{signature} \(\Sigma = (K,arity)\) of a bigraph.  The roots of
the place graph and the link graph, are called \emph{roots} and
\emph{outer names}, respectively.  They form the \emph{outer
  interface} of the bigraph.  The leafs of the place graph and the
link graph, are called \emph{sites} and \emph{inner names},
respectively.  They form the \emph{inner interface}.  A \emph{site}
should be thought of as a hole, which can be filled and properly
connected with another bigraph. 
More formally we define bigraph as:
  
  \begin{definition} [Bigraph] 
    \label{bigraphs_def} 
    A bigraph B under a signature \(\Sigma\) is defined as
    \[B = (V_B, E_B, P_B, ctrl_B, prnt_B, link_B):\langle m,X \rangle
    \to \langle n,Y \rangle, \] where \(m\) and \(n\) are finite
    ordinals that respectively refer to the \emph{sites} and the
    \emph{roots} of the bigraph structure. \(X\) (resp. \(Y,~ V_B,~
    E_B\)) refers to a finite set of \emph{inner names}
    (resp. \emph{outer names, nodes, edges}). \(P_B\) represents a set
    of \emph{ports}.  It is defined as
    \[P_B = \{(v,i)| i \in arity (ctrl_B~ v)\}\] 
    \(ctrl_B : V_B \to K\) assigns controls to nodes.  The \emph{place
      graph} establishes a tree shaped parent ordering among all nodes
    and is defined by \(prnt_B : V_B \cup m \to V_B \cup n\).  The
    relation \(link_B : X \cup P_B \to E_B \cup Y\) maps the union
    of inner names and ports to the union of edges and outer names, which
    represents the hyper-graph called the \emph{link graph}.

    In this definition, \(m, n, X, Y, V_B,\) and \(E_B\) are all
    assumed to be disjoint.  The functions \(arity\), \(ctrl_B\),
    \(prnt_B,\) and \(link_B\) are assumed to be total.

    \(\langle m,X \rangle \) and \( \langle n,Y \rangle\) are called
    the \emph{interfaces} of the bigraph, whereby the former is also
    referred to as the \emph{inner interface} and the latter the
    \emph{outer interface}.
  \end{definition}

  \begin{definition}[Ground Bigraph]
    A bigraph is \emph{ground}, if its inner interface is empty. The
    empty bigraph \(\langle 0,\emptyset \rangle\) is denoted by
    \(\epsilon\).
    
  \end{definition}

  As a running example we use \textsc{ccs} without replication (\(!\))
  and new (\(\nu\)) used to model a vending machine: \(\ov c.co \mid
  c.\ov {co} + c.\ov{t}.\) Here \(c\) represents a coin, \(co\) a cup
  of coffee, and \(t\) a cup of tea. The corresponding bigraph is
  presented in Figure~\ref{fig:example:coffee}, where
  \begin{figure}[t]
    \label{fig:example:coffee}
    \begin{center}
      \includegraphics[scale=0.6]{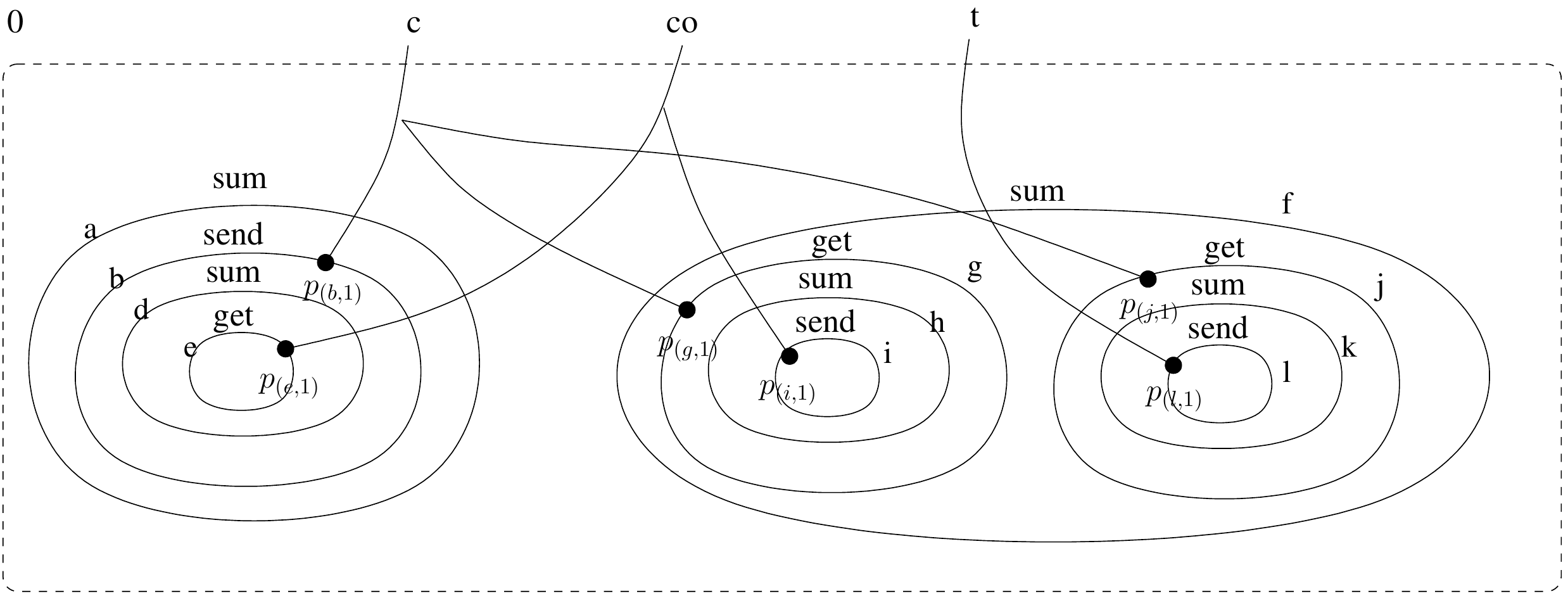} 
    \end{center}
    \caption{Vending  machine example.}
  \end{figure}
  \begin{align*}
    &\Sigma = (\{\ttt{get}, \ttt{send}, \ttt{sum}\},
    \{(\ttt{get},1), (\ttt{send}, 1), (\ttt{sum},0)\}), \\
    &V_B = \{a,b,d,e,f,g,h,i,j,l\}, \\
    & E_B = \emptyset, \\
    &P_B = \{p_{(b,1)}, p_{(e,1)}, p_{(g,1)}, p_{(i,1)}, p_{(j,1)}, p_{(l,1)}\}, \\
    &ctrl_B = \{(a,\ttt{sum}), (b,\ttt{send}), (d,\ttt{sum}), (e,\ttt{get}),(f,\ttt{sum}), (g,\ttt{get}),\\
    &\quad \quad \quad \quad  (h,\ttt{sum}), (i,\ttt{send}), (j,\ttt{get}),(k,\ttt{sum}), (l,\ttt{send})\},  \\
    &prnt_B = \{(a,0), (f,0), (b,a), (d,b), (e,d), (g,f), (h,g), (i,h), (j,f), (k,j), (l,k)\}, \\
    &link_B = \{(p_{(b,1)},c), (p_{(e,1)}, co), (p_{(g,1)}, c), (p_{(i,1)}, co), (p_{(j,1)}, c), (p_{(l,1)}, t), \\
    &m = 0, X = \emptyset, n = 1, \text{ and, } Y = \{c,co,t\}
  \end{align*}
Milner's encoding of \textsc{ccs} in bigraphs requires a so called
\emph{activity} relation for controls, which we omit from this presentation as it can be 
easily added.

  Bigraphs are closed by composition and juxtaposition, which are
  defined as follows:

  \begin{definition}[Composition]
    \label{def:composition}
    Let \(F:\langle k,X \rangle \to \langle m,Y\rangle\)
    and \(G:\langle m, Y \rangle \to \langle n,Z \rangle\) be two
    bigraphs under the same signature \(\Sigma\) with disjoint sets of
    nodes and edges.  The composition \(G \circ F\) is defined
    as:
    \[
    G \circ F = (V, E, P, ctrl, prnt, link) :\langle k,X \rangle \to
    \langle n,Z \rangle
    \]

    where: \(V = V_G \uplus V_F\), \(E = E_G \uplus E_F\), \(ctrl =
    ctrl_G \uplus ctrl_F\)
    \[
    prnt \; x = \left \{ 
    \begin{aligned}
      prnt_F \; x && \text{if \(x \in k \uplus V_F\) and \(prnt_F \; x \in V_F\)}\\
      prnt_G \; j && \text{if \(x \in k \uplus V_F\) and \(prnt_F \; x = j \in m\)}\\
      prnt_G \; x && \text{if \(x \in V_G\)}
    \end{aligned} \right.
    \]
    and 
    \[
    link \; x = \left \{
    \begin{aligned}
      link_F \; x && \text{if \(x \in X \uplus P_F\) and \(link_F \; x \in E_F\)}\\
      link_G \; y && \text{if \(x \in X \uplus P_F\) and \(link_F \; x = y \in Y\)}\\
      link_G \; x && \text{if \(x \in P_G\)}
    \end{aligned} \right.
    \]
  \end{definition}

  Intuitively, \(prnt_{(G \circ F)}\) is defined as the union of
  \(prnt_G\) and \(prnt_F\) where each root \(r\) of \(F\) and each
  site \(s\) of \(G\) such that \(s = r\) satisfies the following: If
  \(prnt_G \; s = y\) and \(prnt_F \; x = r\) then \(prnt_{(G \circ
    F)} \; x = y\) for all $x, y$.  The definition of \(link_{(G \circ
    F)}\) is defined analogously.
  
  \begin{definition}
    [Juxtaposition] Let \(F = (V_F, E_F, P_F, ctrl_F, prnt_F,
    link_F):\langle k,X \rangle \to \langle m,Y\rangle\) and \\ \(G =
    (V_G, E_G, P_G, ctrl_G, prnt_G, link_G):\langle l, W \rangle \to
    \langle n,Z \rangle\) be two bigraphs under the same
    signature \(\Sigma\) that as above have disjoint nodes and edges. The
    \emph{juxtaposed} bigraph \(G \otimes F\) is defined as follows.
      \begin{align*}
        G \otimes F = & (V_F \uplus V_G, E_F \uplus E_G, P_F \uplus P_G,
          ctrl_F \uplus ctrl_G, prnt_F \uplus prnt_G', \\ 
        & \quad link_F \uplus link_G):\langle k + l,X \uplus W \rangle \to
          \langle m + n, X \uplus Z \rangle,
      \end{align*}
    where \(prnt_G' \; (k + i) = m + j\) whenever \(prnt_G \; i = j\).
  \end{definition}

  The composition of two bigraphs can be seen as plugging one bigraph
  structure into the other. The juxtaposition on the other hand can be
  seen as putting two disjoint bigraphs next to each other.

\section{Bigraphical Reactive System}
\label{section:brs}
 
A \textsc{brs} consists of a ground bigraph, which is also called the
\emph{agent} and a set of reaction rules.  In this paper, we discuss
two kinds of reaction rules, those that disallow the matching of sites
(which are said to be \emph{ground}) and those that do not (which are said
to be \emph{parametric}).

\begin{definition}[Ground Reaction Rule]
  A ground rewriting rule consists of two ground bigraphs, the redex
  and the reactum, with the same interfaces \((L: \epsilon \to J,\; R:
  \epsilon \to J)\).
\end{definition}

When we \emph{apply} a ground reaction rule to an agent \(B\), we
require that it can be decomposed into \(B \equiv C \circ L\).  The
result of the application is an agent \(B' \equiv C \circ R\), which
is justified because \(L\) and \(R\) have the same interface.

Parametric reaction rules differ from ground reaction rules by
allowing $L$ and $R$ to contain sites, which may move from one to
another node, be copied, or simply deleted. To this effect parametric
reaction rules define a relation $\eta$ between sites of the reactum
and sites of the redex.  Note the particular direction of
\(\eta\) and the link graph inner interfaces of the redex and the
reactum are empty.

\begin{definition}[Parametric Reaction Rule] \label{def:prr} A
  parametric reaction rule is a triple of two bigraphs and a total
  function from \(R\)'s sites to \(L\)'s sites \(\eta\)

  \[
      (L: \langle m, X \rangle \to J,\; R: \langle m', X \rangle \to J, \; \eta : m' \to m)
  \]
  \end{definition}

  The semantics of parametric \textsc{brs} is also based on
  decomposition of the agent, where an instantiation function is
  deduced from \(\eta\) to compute the new tails of the bigraph
  structure. This function is defined up to an equivalence. The
  \(\bumpeq\)-equivalence is defined as follows:
  
  \begin{definition} [Equivalence]
    Let \(B\) and \(G\) be two bigraphs with the same interface
    \(\langle k,X \rangle \to \langle m, Y\rangle \).  \(B\) and \(G\)
    are called \(\bumpeq\)-equivalent (\emph{lean} equivalent) if
    there exist two bijections \(\rho_V:V_B \to V_G\) and
    \(\rho_E:E_B \to E_G\) that respect the structure, in the
    following sense:
    \begin{itemize}
    \item \(\rho\) preserve controls, i.e. \(ctrl_G \circ \rho_V =
      ctrl_B\), and therefore induces a bijection \(\rho_P:P_F \to P_G\)
      defined as \(\rho_P \; (v,i) = ((\rho_V \; v), i)\).\\
    \item \(\rho\) commutes with the structural maps as follow: 
      \[\begin{array}{l}
      prnt_G \circ (Id_m \uplus \rho_V) = (Id_n \uplus \rho_V) \circ
      prnt_B\\ 
      link_G \circ (Id_X \uplus \rho_P) = (Id_Y \uplus \rho_E) \circ
      link_B.
      \end{array}
      \]
    \end{itemize}
  \end{definition}

  \begin{definition} [Instantiation]
    \label{def:ov_eta}
    Let \(F \equiv G \circ (d_0 \otimes \dots \otimes d_{m-1}):\langle
    n,Y\rangle \to \langle m, X\rangle\) be a bigraph where \(G\) is a
    link graph, the \(d\)'s have no inner names and a unary outer face
    and \(\eta:m' \to m\) the relation on sites from
    Definition~\ref{def:prr}.  We refer to an instantiation of $\eta$
    on \(F\) as \(\ov \eta\) that is defined as follows: \[
      \ov \eta \; F = G \circ (d_0' \Vert \dots \Vert d_{m'-1}')
    \] where \(\forall j \in m', d_j' \bumpeq d_{(\eta \;
    j)}\). Following \cite{MILNER09} we write $d \Vert d'$ for $d
  \otimes d'$ where $d$ and $d'$ can share edges and outer names.
  \end{definition}
  
  Let \((B, \mathcal{R})\) be a {\brs}, and \((L, R,\eta) \in
  \mathcal{R}\) the parameterised reaction rule.  If \(B \equiv C \circ
  (D \otimes L) \circ C'\) then we can apply the reaction rule
  $\mathcal{R}$ and rewrite the bigraph \(B\) into \(B' \equiv C
  \circ (D \otimes R) \circ (\ov \eta \; C')\).
  \begin{figure}[t]
    \label{fig:ccs:transition_rule}
    \begin{center}
      \includegraphics[scale=0.55]{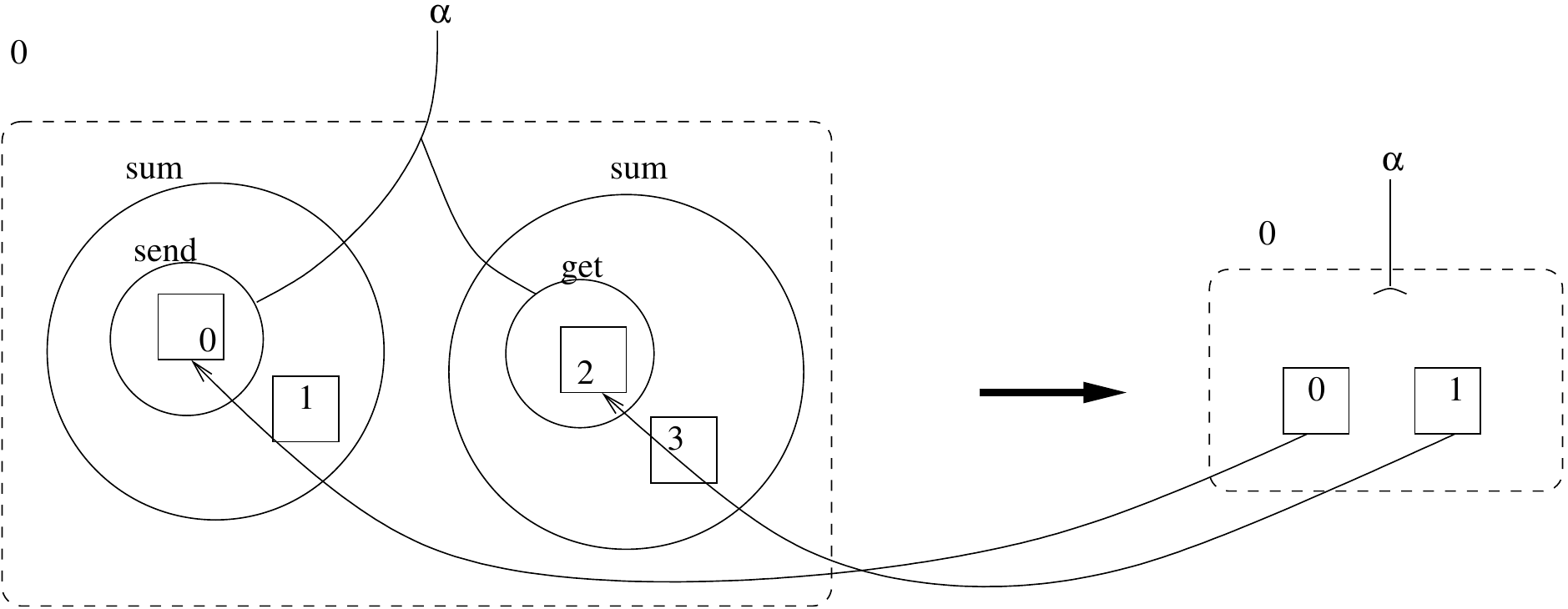}
    \end{center}
    \caption{The bigraph encoding for the \textsc{ccs} \(\tau\)-transition rule.}
  \end{figure}
  \begin{figure}[t]
    \begin{center}
      \includegraphics[scale=0.55]{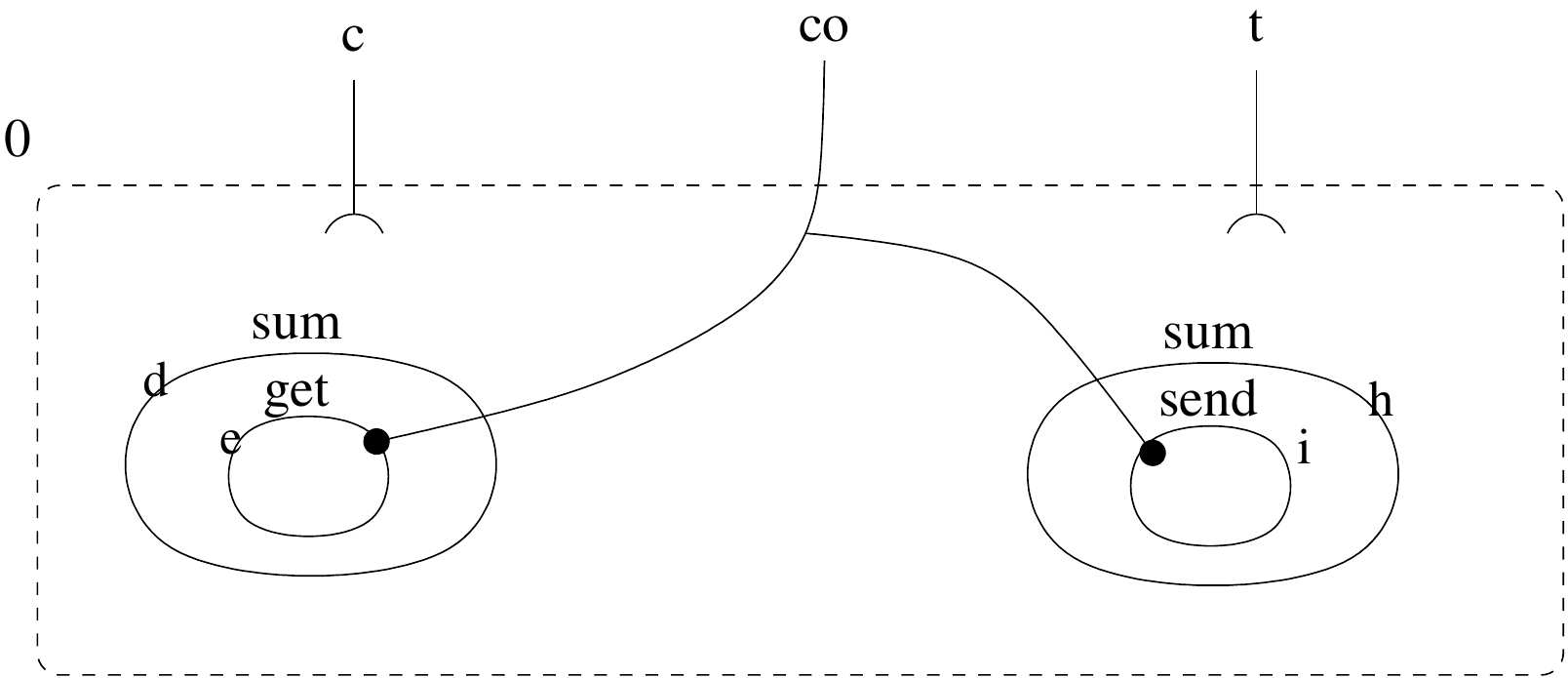}
    \end{center}
    \caption{The \textsc{ccs} term \(co \mid \ov {co}\).}
    \label{fig:ccs:step}
  \end{figure}
  For instance, the \(\tau\)-reaction rule the Milner's encoding of
  \textsc{ccs} in bigraph, \( (\alpha.P + P') \mid (\ov \alpha.Q + Q')
  \rightarrow^\tau Q \mid P\) is encoded as parametric reaction rule
  presented in Figure~\ref{fig:ccs:transition_rule}, in this case
  \(\eta=\{(0,0), (1,2)\}\). And using this rule, the coffee machin
  example could reduce in one step in the bigraph represented in
  Firgure~\ref{fig:ccs:step}.

\section{Bigraph Relational Model }

\label{section:encoding_bigraph}
Next, we tackle the definition of the bigraph relational model that
arises form the graphical presentation of the categorical model.  In
our model, bigraphs are identified by name.  The disjoint sets of
names \(n\), \(V_B\), \(m\), \(P_B\), \(Y\), \(X\), \(E_B\), \(K\) are
expressed as \emph{base} relation symbols \ttt{bigraph}, \ttt{root},
\ttt{node}, \ttt{site}, \ttt{port}, \ttt{o\_name}, \ttt{i\_name},
\ttt{e\_name}, \ttt{control}.  The \(arity\) function is encoded as
\ttt{arity}, a ternary relation symbol indexed by \ttt{control}, a
natural number, and a \ttt{bigraph}.  And similarly, relations
\(prnt_B\), \(link_B\), \(ctrl_B\), and \(P_B\) are encoded by the 
following \emph{operational} relation symbols:
\begin{itemize}
\item \ttt{prnt S D B} where \ttt{S} $\in$ \ttt{node} $\cup$ 
  \ttt{site} and \ttt{D} $\in$ \ttt{node} $\cup$ \ttt{root},
\item \ttt{link S D B} where \ttt{S} $\in$ \ttt{i\_name} $\cup$
  \ttt{port} and \ttt{D} $\in$ \ttt{o\_name} $\cup$
  \ttt{e\_name},
\item \ttt{lc A K B} where \ttt{A} $\in$ \ttt{node} and \ttt{K} $\in$
  \ttt{control},
\item \ttt{lp P A B} where \ttt{P} $\in$ \ttt{port} and \ttt{A} $\in$
  \ttt{node},
\end{itemize}
where the symbol \ttt B represents the name of the bigraph \(B\)
  and where disjunctive types are encoded as relations:
\begin{itemize}
\item \ttt{prnt\_src}, with \ttt{src\_n} for nodes and
  \ttt{src\_s} for sites as constructors,
\item \ttt{prnt\_dst}, with \ttt{dst\_n} for nodes and \ttt{dst\_r}
  for roots,
\item \ttt{link\_src}, with \ttt{src\_i} for inner names and
  \ttt{src\_p} for ports,
\item \ttt{link\_dst}, with \ttt{dst\_o} for outer names and
  \ttt{dst\_e} for edges.
\end{itemize}
We declare the following \emph{structural} relation
symbols as well:
\begin{itemize}
\item \ttt{is\_root R B} where \ttt{R} $\in$ \ttt{root} and \ttt{R} belongs to \ttt{B},
\item \ttt{is\_node N B} where \ttt{N} $\in$ \ttt{node} and \ttt{N} belongs to \ttt{B},
\item \ttt{is\_site S B} where \ttt{S} $\in$ \ttt{site} and \ttt{S} belongs to \ttt{B},
\item \ttt{is\_port P B} where \ttt{P} $\in$ \ttt{port} and \ttt{P} belongs to \ttt{B},
\item \ttt{is\_o\_name O B} where \ttt{O} $\in$ \ttt{o\_name} and \ttt{O} belongs to \ttt{B},
\item \ttt{is\_i\_name I B} where \ttt{I} $\in$ \ttt{i\_name} and \ttt{I} belongs to \ttt{B},
\item \ttt{is\_e\_name E B} where \ttt{E} $\in$ \ttt{e\_name} and \ttt{E} belongs to \ttt{B}.
\end{itemize}


\setcounter{equation}{0}
\begin{figure}[t]
\label{fig:setrwr:valid}
\begin{align*}
    & \text{base cases}\\
    &\tag{dr}\quad\{\ttt{is\_root R, has\_child\_p (dst\_r R) z}\} 
    \uplus \Delta \mapsto \Delta\\
    &\tag{do}\quad\{\ttt{is\_o\_name O, has\_child\_l (dst\_o O) z}\} \uplus \Delta  
    \mapsto \Delta\\
    &\tag{de}\quad\{\ttt{is\_e\_name E, has\_child\_l (dst\_e E) z}\} \uplus \Delta, 
    \mapsto \Delta\\
    & \text{recursive cases}\\
    & \tag{lgpsz}\quad\{\ttt{is\_port P, lp P A, vp A (s z), 
      link (src\_p P) D,}\\
    & \qquad \ttt{has\_child\_l D (s N)}\} \uplus \Delta \mapsto \{\ttt{has\_child\_l D N}\} \uplus \Delta\\
    & \tag{lgi}\quad\{\ttt{is\_i\_name I, link (src\_i I) D, has\_child\_l D (s N)}\}\\
    & \quad \quad \uplus \Delta \mapsto \{\ttt{has\_child\_l D N}\} \uplus \Delta\\
    & \tag{lgs}\quad\{\ttt{is\_site S, prnt (src\_s S) D, 
      has\_child\_p D (s N)}\} \uplus \Delta\\
    & \quad \quad \mapsto \{\ttt{has\_child\_p D N}\} \uplus \Delta\\
    & \tag{pgnz}\quad\{\ttt{is\_node A, has\_child\_p (dst\_n A) z, prnt (src\_n A) D,} \\
    &  \quad \quad \ttt{has\_child\_p D (s N), lc A K}\} \uplus \Delta
    \mapsto \{\ttt{has\_child\_p D N}\} \uplus \Delta & \text{if} \; arity \; 
    \ttt K = 0\\
    &  \tag{pgns}\quad\{\ttt{is\_node A, has\_child\_p (dst\_n A) z, prnt (src\_n A) D,}\\
    &  \quad \quad \ttt{has\_child\_p D (s N), lc A K}\} \uplus \Delta\\
    &  \qquad \quad \mapsto \{\ttt{has\_child\_p D N, vp A N'}\} \uplus \Delta & 
    \text{if} \; arity \; \ttt K = N' > 0\\
    & \tag{lgps}\quad\{\ttt{is\_port P, lp P A, 
      vp A (s (s N')),} \\
    &  \quad \quad \ttt{link (src\_p P) D, has\_child\_l D (s N)}\} \uplus \Delta\\
    & \quad \quad \mapsto \{\ttt{vp A (s N'), has\_child\_l D N}\} \uplus \Delta
  \end{align*}
\caption{Bigraph Validity}
\label{fig:validity}
\end{figure}

Next, we sketch an algorithm for deciding the validity of a bigraph in
the relational model.  The algorithm is deceptively simple: using the
rules depicted in Figure~\ref{fig:validity}, we rewrite the encoding
of a bigraph to the empty set by checking the validity of the place
graph and the link graph and the valid use of control's arity.  As we
will show below, the algorithm is confluent and strongly normalising.
The rules are partially sequentialised in such a way that children are
rewritten before the parents and nodes before ports. We make this
information explicit and define three more operational relational
symbols one for the place graph, an other one for the link graph, and
one for the control's arity.
\begin{itemize}
\item \ttt{has\_child\_p D N B} where \ttt{D} $\in$ \ttt{node} $\cup$
  \ttt{root}, \ttt{N} a natural number, and \ttt B a bigraph name.
\item \ttt{has\_child\_l D N B} where \ttt{D} $\in$ \ttt{e\_name}
  $\cup$ \ttt{o\_name}, \ttt{N} a natural number, and \ttt B a bigraph
  name.
\item \ttt{vp A N B} where \ttt{A} \(\in\) \ttt{node}, \ttt{N} a
  natural number, and \ttt{B} a bigraph name.
\end{itemize}

\begin{figure}[t]
{\small \input{encoding}}
\caption{Encoding function from a bigraph structure.}
\label{fig:encoding}
\end{figure}

The encoding of bigraph B is now straightforward. It is defined as the
multi-set \(S = \llb B \rrb\) in
Figure~\ref{fig:encoding}.

For reasons of convenience, we omit bigraph names from relations, and
we use uppercase characters for logic variables in rewrite
rules. Furthermore we use lower case \ttt{z} and \ttt{s} for zero and
successor.

We show that the rewriting system is strongly normalising and
implements a decision procedure for checking the validity of bigraphs.

  \begin{lemma} [SN]
    \label{lm:setrwr:normalising} 
    This multi-set rewriting system is strongly normalising for any finite multi-set \(S\).
  \end{lemma}

  \begin{proof}
    By a trivial induction of the size of the set \(S\).
  \end{proof}
  In the following we write \(S \longrightarrow S'\) for transitive
  closure of $\mapsto$ from Figure~\ref{fig:validity}.  We say $S$ is
  \emph{valid} if and only if $S \longrightarrow \emptyset$, $S$
  contains only unique elements and operational symbols
  \ttt{has\_child\_p}, \ttt{has\_child\_l} and \ttt{vp} are unique on
  their first argument, respectively \(\ttt{node} \cup \ttt{root}\),
  \(\ttt{o\_name} \cup \ttt{e\_name}\) and \ttt{port}.

  \begin{example}The bigraph $B$ from
    Figure~\ref{fig:example:coffee}, is represented as follows:
    {\small
    \begin{align*}
      \llb B \rrb &= \{\ttt{is\_root 0 B}, \;\ttt{has\_child\_p 0 (s (s z)) B}\\
      & \ttt{is\_node a B}, \;\ttt{lc a get B}, \;\ttt{has\_child\_p a (s z) B}, \\
      & \ttt{is\_node b B}, \;\ttt{lc b send B}, \;\ttt{has\_child\_p b (sz) B}, \\
      & \ttt{is\_node d B}, \;\ttt{lc d sum B}, \;\ttt{has\_child\_p d (s z) B}, \\
      & \ttt{is\_node e B}, \;\ttt{lc e get B}, \;\ttt{has\_chidl\_p e z B}, \\
      & \ttt{is\_node f B}, \;\ttt{lc f sum B}, \;\ttt{has\_child\_p f (s (s z)) B}, \\
      & \ttt{is\_node g B}, \;\ttt{lc g get B}, \;\ttt{has\_child\_p g (s z)}, \\
      & \ttt{is\_node h B}, \;\ttt{lc h sum B}, \;\ttt{has\_child\_p h (s z) B}, \\
      & \ttt{is\_node i B}, \;\ttt{lc i send B}, \;\ttt{has\_child\_i g z}, \\
      & \ttt{is\_node j B}, \;\ttt{lc j get B}, \;\ttt{has\_child\_p d (s z) B}, \\
      & \ttt{is\_node k B}, \;\ttt{lc k sum B}, \;\ttt{has\_chidl\_p k z B}, \\
      & \ttt{is\_node l B}, \;\ttt{lc l send B}, \;\ttt{has\_child\_p l z B}, \\
      & \ttt{is\_port \(p_{(b,1)}\) B}, \;\ttt{lp \(p_{(b,1)}\) b B}, \;
        \ttt{is\_port \(p_{(e,'1)}\) B}, \;\ttt{lp \(p_{(e,1)}\) e B},\\
      & \ttt{is\_port \(p_{(g,1)}\) B}, \;\ttt{lp \(p_{(g,1)}\) g B}, \;
        \ttt{is\_port \(p_{(i,1)}\) B}, \;\ttt{lp \(p_{(i,1)}\) i B},\\
      & \ttt{is\_port \(p_{(j,1)}\) B}, \;\ttt{lp \(p_{(j,1)}\) j B},\;
        \ttt{is\_port \(p_{(l,1)}\) B}, \;\ttt{lp \(p_{(l,1)}\) l B},\\
      & \ttt{is\_o\_name c B}, \;\ttt{has\_child\_l c (s (s (s z))) B},\\
      & \ttt{is\_o\_name co B}, \;\ttt{has\_child\_l co (s (s z)) B},\\
      & \ttt{is\_o\_name t B}, \;\ttt{has\_child\_l t (s z) B},\\
      & \ttt{prnt a 0 B}, \;\ttt{prnt f 0 B}, \;\ttt{prnt b a B},\;
        \ttt{prnt d b B}, \;\ttt{prnt e d B}, \;\ttt{prnt g f B}, \\
      & \ttt{prnt h g B}, \;\ttt{prnt i h B}, \;\ttt{prnt j f B},
        \ttt{prnt k j B}, \;\ttt{prnt l k B},\\
      & \ttt{link \(p_{(b,1)}\) c B}, \;\ttt{link \(p_{(e,1)}\) c B}, \;\ttt{link \(p_{(g,1)}\) c B},\\ 
      & \ttt{link \(p_{(i,1)}\) co B}, \;\ttt{link \(p_{(j,1)}\) co B}, \;\ttt{link \(p_{(l,1)}\) t B}\}
    \end{align*}}
  \end{example}

  \begin{lemma}
    \label{lm:setrwr:functional}
    Let $S$ be valid. Then the relations in $S$ defined by the operational
    symbols are total, acyclic, and single valued.
  \end{lemma}

  \begin{lemma}
    \label{lm:setrwr:has_child}
    Let $S$ be valid.
      \begin{itemize}
      \item \(\forall y \in\) \ttt{\upshape node} \(\cup\) \ttt{\upshape
        root}, \(\forall N\) a natural number, \ttt{\upshape
        has\_child\_p y N} \(\in S\) implies \(| \{x \mid
        \ttt{\upshape prnt}\;x\; y \in S\}| = N\)\\
      \item \(\forall y \in\) \ttt{\upshape edge} \(\cup\) \ttt{\upshape
        o\_name}, \(\forall N\) a natural number, \ttt{\upshape
        has\_child\_l y N} \(\in S\) implies \(|\{x \mid
        \ttt{\upshape link} \; x \; y \in S\} | = N\)\\
      \item \(\forall v \in\) \ttt{\upshape node}, \(\forall N\) a natural number,
        \ttt{\upshape vp v N} \(\in S\) implies \(|\{x \mid \ttt{\upshape lp x v}
        \in S\}| = N\).
      \end{itemize}
  \end{lemma}
  
  \iftechreport
\else{
    The proofs of Lemma~\ref{lm:setrwr:functional}
    and~\ref{lm:setrwr:has_child} rely on the definition of multi-set
    rewriting consuming elements of the multi-set.}\fi

  \begin{lemma}[Normal Form]
    \label{lm:setrwr:confluent}
    Let $S$ be valid.  Then $\emptyset$ is the unique normal form of
    $S$ with respect to  $\longrightarrow$.
  \end{lemma}

  \begin{proof}
    Since strong normalisation has already been proven in
    Lemma~\ref{lm:setrwr:normalising}, we prove local confluence in
    order to apply the Newman's lemma.  
    \iftechreport
\else{ The critical pairs that respect the hypothesis
      are either trivial or join-able in one step.  }\fi
  \end{proof}

  By induction over the rewrite trace, we can easily convince
  ourselves that \({\mathit{valid}} \; S\) holds if and only if $S$ encodes
  a bigraph.

  \begin{theorem} [Inversion]
    \label{th:setrwr:bigraph}
    If $S$ is valid then there exists a bigraph B, s.t.\ \(\llb B \rrb =
    S\).
  \end{theorem}

  \begin{proof}
    Lemma~\ref{lm:setrwr:functional} and~\ref{lm:setrwr:has_child}
    ensure that the set holds the properties of the graphical
    definition of a bigraph.
  \end{proof}

  \begin{figure}[t]
    \input{interpretation}
    \caption{Interpretation function.}
    \label{fig:interpretation}
  \end{figure}

  This theorem guarantees that the encoding as defined in
  Figure~\ref{fig:encoding} has an inverse (which is only defined on
  valid sets).  It is
  defined in Figure~\ref{fig:interpretation} and for which we write
  \(\llb S \rrb^\star\).  Furthermore, we have shown that there exists a
  bijection between bigraphs and their representations as valid
  multi-sets.

  \begin{theorem} [Adequacy]
    \label{th:adequacy}
    Let B be a bigraph under a signature \(\Sigma=(K, arity)\), then
    \(\big\llb \llb B \rrb \big\rrb^{\star} \bumpeq B\).
  \end{theorem}

  Next, we show that composition and juxtaposition of bigraphs
  are provided ``for free'' in the bigraph relational model.  They
  basically correspond to multi-set union.

  Let $C$ be a bigraph.  Next, we partition \(\llb C \rrb\) into three parts,
  \(\widetilde{C} \uplus out_C \uplus in_C\) where  
  \begin{enumerate}
  \item \(out_C\) contains only references to roots \ttt{is\_root},
    outer names \ttt{is\_o\_name}, place graph parent relations
    \ttt{prnt} on roots and link graph parent relations \ttt{link} on
    outer names,
  \item \(in_C\) contains only references to sites \ttt{i\_site},
    inner names \ttt{is\_i\_name}, place graph parent relations
    \ttt{prnt} on sites and link graph parent relations \ttt{link} on
    inner names,
  \item and \(\widetilde{C}= \llb C \rrb \setminus out_C \setminus in_C \).
  \end{enumerate}

  Also, let $B$ be a bigraph such that \(B = C \circ C'\) for two bigraphs,
  \(C\) and \(C'\).  We define the set \(eq_{CC'}\) as follows :
    \begin{align*}
    eq_{CC'} =& \{ \ttt{prnt x y} \mid \ttt{is\_root r} \in \llb C' \rrb 
       \wedge \ttt{is\_site s} \in \llb C \rrb\\
    & \quad \wedge \mid \ttt s \mid = \mid \ttt r \mid {}\wedge{} \ttt{prnt x r} 
       \in \llb C' \rrb \wedge \ttt{prnt s y} \in \llb C \rrb\}\\
    & \cup \{ \ttt{link x y} \mid \ttt{is\_i\_name i} \in \llb C \rrb \wedge 
       \ttt{is\_o\_name o} \in \llb C' \rrb\\
    & \quad \wedge \mid \ttt i \mid = \mid \ttt o \mid {}\wedge{} \ttt{link i y} \in \llb C 
       \rrb \wedge \ttt{link x o} \in \llb C' \rrb\}
    \end{align*}
  Note, that \(eq_{CC'}\) is built from \(out_{C'}\) and \(in_C\).


  \begin{lemma}
    \label{lm:brs_encoding:composition}
    Let $B$ be a bigraph.  If \(B \equiv C \circ C'\) then \(\llb B
    \rrb = out_c \; \cup \; \widetilde{C} \; \cup \; eq_{CC'} \;
    \cup \; \widetilde{C'} \; \cup \; in_{C'}\).
  \end{lemma}

  \begin{lemma}
    \label{lm:brs_encoding:juxtaposition}
    Let $B$ be a bigraph. If \(B \equiv C \otimes C'\) then \(\llb B
    \rrb = \; \widetilde{C} \; \cup \; \widetilde{C'} \; \cup \;
    out_C \; \cup \; in_{C} \;\cup \; out_{C'} \; \cup \; in_{C'} =
    \llb C \rrb \; \cup \; \llb C' \rrb\) .
 \end{lemma}
  \iftechreport
\else{ 
     
    The proofs of these lemmas are direct, using the definitions of
    composition and juxtaposition, in particular, the disjointness
    property of \(C\) and \(C'\).  }\fi
  
\section{Modelling Reaction Rules}
\label{section:encoding_brs}
 
In this section, we illustrate how we model reaction rules.

    
    
  \subsection{Ground Reaction Rule}

  Recall that applying a ground reaction rule
  \((L, R)\) to an agent $B$ proceeds by decomposing \(B\) into \(C
  \circ L\) for some bigraph \(C\) and then replacing \(L\) by \( R\).
  Therefore, in the model, we only need to partition the agent \(\llb
  B \rrb\) into two sets with respect to \(C\), one that is affected
  by the reaction rule (here \(\widetilde{L} \; \cup \; eq_{CL}\)) and
  the other that is not.
  \[
    \llb C \circ L \rrb = out_C \; \cup \; \widetilde{C} \; \cup
      \; eq_{CL} \; \cup \; \widetilde{L} \mapsto
    \llb C \circ R \rrb = out_C \; \cup \widetilde{C} \; \cup
      \; eq_{CR} \; \cup \; \widetilde{R}
  \]

  For a given bigraph \(C\), we can think of a ground reaction
  rule as a multi-set rewrite rule that replaces among other things the
  set $eq_{CL}$ by $eq_{CR}$:
  \[
   \widetilde{L} \; \cup \; eq_{CL} \mapsto
      \widetilde{R} \; \cup \; eq_{CR}
  \]

  We can rid this rule of the dependency on $C$.  We know, first, that
  the inner interface of $C$ must be the same as the outer interface
  of $L$ (and therefore also $R$).  Second, the components in the
  interface that depend on $C$ are only roots and outer names.
  Therefore, instead of quantifying over $C$, we can reformulate the
  rule by simply quantifying over the aforementioned components. 

  \subsection{Parametric Reaction Rule}
  \label{section:encoding_parametric_brs}


  The parametric reaction rule from Definition~\ref{def:prr} is a
  triple that consists of two bigraphs, and a function \(\eta\) that
  maps sites from the reactum to sites in the redex.

  A parametric reaction rule is applied to an agent $B$ if $B$ can be
  decomposed into \(C \circ (D \otimes L) \circ C'\) where \(C, D\)
  and \(C'\) are bigraphs and \(L\) is the redex. The result of the
  application is the agent \(C \circ (D \otimes R) \circ (\ov \eta \;
  C')\) where \(\ov \eta\) is defined in Definition~\ref{def:ov_eta}.
  The basic idea is essentially the same as in the ground case,
  therefore we proceed analogously and model decomposition by partition

 \begin{align*}
    \llb C \circ (D \otimes L) \circ C' \rrb &= out_C \; \cup \;
    \widetilde{C} \; \cup \; eq_{C((D \otimes L) \circ C')} \; \cup \;
    \widetilde{((D \otimes L)\circ C')} \; \cup \; eq_{LC'} \\
    & \qquad \; \cup \; \widetilde{C'} \; \cup \; in_{((D \otimes L)\circ C')}\\
      &= out_C \; \cup \; \widetilde{C} \; \cup \; eq_{C((D \otimes L) \circ
        C')} \; \cup \; \widetilde{L} \; \cup \; \widetilde{D} \; \cup \;
      eq_{DC'} \; \cup \; eq_{LC'} \; \cup \;  \widetilde{C'}
  \end{align*}
  and model  parametric reaction rule as a multi-set
  rewriting rule. 

  \begin{align*}
    &eq_{C((D \otimes L) \circ C')} \; \cup \; \widetilde{L} \; \cup \;
    \widetilde{D} \; \cup \;
    eq_{DC'} \; \cup \; eq_{LC'} \; \cup \;  \widetilde{C'}\\
    &\quad \mapsto eq_{C((D \otimes R)\circ (\ov \eta \; C'))} \; \cup \;
    \widetilde{R} \; \cup \; \widetilde{D} \; \cup \; eq_{D(\ov \eta \;
      C')} \; \cup \; eq_{L(\ov \eta \; C')} \; \cup \; \widetilde{( \ov
      \eta \; C')}
  \end{align*}
 
  Differently from above, the formulation of the rule is not only
  dependent on $C$, but also on $C'$.  This time, however things are
  not as direct, in part because the inner interfaces of the redex and
  the reactum do not match.  In Definition~\ref{def:prr}, we use
  \(\eta\) to coerce one to the other, which means that the interfaces
  between $(\ov \eta\;C')$ and $R$  actually do match.

  Applying \(\ov \eta\) to \(C'\) is algorithmically simple: on the
  place graph, the operation recursively \emph{deletes} all sites that
  are not in the range of $\eta$, \emph{moves} all sites that have a
  unique image under $\eta$, or \emph{copies} all sites that do not
  have a unique image under $\eta$; on the link graph, it only
  administers links from and to ports (as proposed
  in~\cite{MILNER09}).  The computational essence of these operations
  is captured in terms of a few multi-set rewriting rules, that
  iterates over the place graph, which we discuss in more detail in
  Section~\ref{section:implementation}.

\subsection{Meta Theory}
We show that modelling ground and the parametric reduction rules as
rewrite rules is sound and complete.

  \begin{theorem} [Soundness]
    \label{th:brs_encoding:ground_soundness}
    Let \( B, B'\) be agents, \(\mathcal{R}\) a set of reaction rules
    and \((L, R, \eta) \in \mathcal{R}\).  If \(B\) can be rewritten
    into \(B'\) by \(L, R, \eta\) and \(W \mapsto Z\) is the
    corresponding rewriting system, then the following diagram
    commutes:
    
    \begin{center}
    \begin{tikzpicture}[normal line/.style={->},font=\scriptsize]
      \matrix (m) [matrix of math nodes, row sep=4em,
      column sep=12em]
      {
        B & B'\\
        \llb B \rrb & \llb B' \rrb\\  
      };
      \path[normal line]
      (m-1-1) edge node[auto] {\((L, R, \eta)\)} (m-1-2)
        edge node[auto,swap]  {\( \llb \cdot\rrb\)} (m-2-1)
      (m-1-2) edge node[auto]  {\(\llb \cdot \rrb\)} (m-2-2)
      (m-2-1) edge node[auto]  {\(W \mapsto Z\)} (m-2-2);
    \end{tikzpicture}
    \end{center}
  \end{theorem}
    
  \begin{proof}
    \emph{Ground reaction rule.} \(L\) and \(R\) are ground and the
    graph of \(\eta\) is empty.  \(B \equiv C \circ L\) and \(B'
    \equiv C \circ R\), by Lemma~\ref{lm:brs_encoding:composition} and
    because interfaces of \(L\) and \(R\) are the same, \(\llb B \rrb
    =\widetilde{C} \; \cup \; \widetilde{L} \; \cup \; out_C
    \; \cup \; in_{L} \; \cup \; eq_{CL}\) and \(\llb B' \rrb
    =\widetilde{C} \; \cup \; \widetilde{R} \; \cup \; out_C
    \; \cup \; in_{L} \; \cup \; eq_{CR}\), therefore \(\llb B' \rrb\)
    is the result of one step of the multi-rewriting system
    \(\widetilde{L} \; \cup \; eq_{CL} \to \widetilde{R} \;
    \cup \; eq_{CR}\) applied on \(\llb B \rrb\).  \emph{Parametric
      reaction rule.} Analogous.
  \end{proof}

  Conversely, all multi-set rewriting system that encode a ground
  reaction rule respects the semantics of ground \textsc{brs}.
  
  \begin{theorem}[Completeness]
    \label{th:brs_encoding:ground_completeness}
    Let \(X, Y\) be valid sets. Furthermore let \(W \mapsto Z\) be sets
    such that there exists an \((L, R, \eta)\) that is a reaction rule
    that can be applied on \(\llb X \rrb^{\star}\) and \(W \mapsto Z\) is
    the modeled reaction rule of \((L,R,\eta)\), then the following
    diagram commutes:
    
    \begin{center}
    \begin{tikzpicture}[normal line/.style={->},font=\scriptsize]
      \matrix (m) [matrix of math nodes, row sep=4em,
      column sep=7em]
      {
        X & Y\\
        \llb X \rrb^\star & \llb Y \rrb^\star\\  
      };
      \path[normal line]
      (m-1-1) edge node[auto] {\(W \mapsto Z\)} (m-1-2)
        edge node[auto,swap]  {\(\llb \cdot \rrb^\star\)} (m-2-1)
      (m-1-2) edge node[auto]  {\(\llb \cdot \rrb^\star\)} (m-2-2)
      (m-2-1) edge node[auto]  {\((L, R, \eta)\)} (m-2-2);
    \end{tikzpicture}
    \end{center}
  \end{theorem}

  \begin{proof}   \emph{Ground reaction rule.}
    Using Lemma~\ref{lm:brs_encoding:composition} and
    Theorem~\ref{th:adequacy}, \(\widetilde{L} \; \cup \;
    eq_{XL} \in X\) implies that \(\llb X \rrb^{\star} \equiv C \circ
    L\). Therefore \(\widetilde{R} \; \cup \; eq_{XR} \in X\)
    implies that \(\llb Y \rrb^{\star} \equiv C \circ R\).
    \emph{Parametric reaction rule.} Analogously.
  \end{proof}

\section{Implementation in {\celf}}
\label{section:implementation}

We turn now to the original motivation of this work and evaluate the
bigraph relational model empirically.  The very nature of the rules
depicted in Figure~\ref{fig:validity} suggests a language based on
multi-set rewriting, such as Maude, Elan, $\lambda$Prolog, CHR, or
Celf.  Because of Celf's features, in particular linearity and
higher-order abstract syntax, we have decided to use Celf as our
implementation platform.

And indeed, the implementation of the bigraph relational model is
straightforward.  Roots, nodes, sites, etc., are encoded using Celf's
intuitionistic features, and the evidence that something is a root, a
parent, or a port is captured by linear assumptions using dependent
types.  Consequently a bigraph is represented by the Celf context.
The multi-set rewriting rules as depicted in Figure~\ref{fig:validity}
and the reaction rules from Section~\ref{section:encoding_brs} are
encoded using linear types and the concurrency modality.  For example, the
rewrite rule 
\[\{\ttt{is\_root R, has\_child\_p (dst\_r R) z}\} 
    \uplus \Delta \mapsto \Delta
\]
is implemented in {\celf} as a constant 
\begin{verbatim}
          dr : is_root R B * has_child_p (dst_r R) z B -o {1}.
\end{verbatim}
where \verb!is_root! carries a reference to the bigraph it is a root
for, and all uppercase variables are implicitly $\Pi$ quantified.  Celf
provides a sophisticated type inference algorithm that infers all
omitted types (or terminates with an error if those cannot be
found).

Celf also comes with a forward directed logic programming engine in
the style of Lollimon~\cite{Lopez05ppdp}, which resembles the CHR
evaluation engine.  During operation, the uppercase variable names are
replaced by logic variables, which are subsequently instantiated by
unification if the rule is applied.  Note that the properties of the
encoded bigraph reactive system are preserved: If the reaction rules
are strongly normalising then so is their encoding.  As an
illustration of our experiments we depict an encoding of bigraph
validity (see Figure~\ref{fig:setrwr:valid}) as a type family
\ttt{valid} and a few Celf declarations in Figure~\ref{fig:clf:valid}.
Note how similar the two figures are. 

\begin{figure}[t]
\label{fig:clf:valid}
\begin{center}
\begin{small}
\begin{verbatim}
%% base cases
dr : is_root R B * has_child_p (dst_r R) z B -o {1}.
do : is_o_name O B * has_child_l (dst_o O) z B -o {1}.
de : is_e_name E B * has_child_l (dst_e E) z B -o {1}.

%% recursive cases
lgpsz : is_port P Bi * lp P A Bi * vp A (s z) Bi * link (src_p P) D Bi
  * has_child_l D (s N) Bi -o {has_child_l D N Bi}.
lgi : is_i_name I Bi * link (src_i I) D Bi
  * has_child_l D (s N) Bi -o {has_child_l D N Bi}.
pgs : is_site S Bi * prnt (src_s S) D Bi * has_child_p D (s N) Bi
  -o {has_child_p D N Bi }.
pgnz : is_node A Bi * has_child_p (dst_n A) z Bi * prnt (src_n A) D Bi
  * has_child_p D (s N) Bi * lc A K Bi -o arity K z -> {has_child_p D N Bi}.
pgns : is_node A Bi * has_child_p (dst_n A) z Bi * prnt (src_n A) D Bi
  * has_child_p D (s N) Bi * lc A K Bi -o arity K (s N')
    -> {has_child_p D N Bi * vp A (s N') Bi}.
lgps : is_port P Bi * lp P A Bi * vp A (s (s  N)) Bi
  * link (src_p P) D Bi * has_child_l D (s N') Bi
    -o {vp A (s  N) Bi * has_child_l D N' Bi}.
\end{verbatim}
\end{small}
\end{center}
\caption{The implementation of the valid relation in {\celf}.}
\end{figure}

The higher-order nature of \celf{} allows us to express rewriting
rules that dynamically introduce new rewriting rules on the fly.  In
Celf, rewriting rules are first-class citizens.  The logical principle
behind this technique is called \emph{embedded implications}.  By
nesting them we achieve elegant encodings.

An example is the encoding of a parametric bigraph reaction rule \((L,
R, \ov \eta)\). The definition of the Celf signature is rather
involved, where we use linearity and token system (\(\llb \emptyset,
\emptyset \rrb\)) in order to sequentialise the reaction rule.  Below
we give an algorithm that computes the Celf declaration
\[
\ttt{rule}_{(L, R, \ov \eta)}: \llb(m,\eta)\rrb \otimes \big ( \llb
(\emptyset, \emptyset) \rrb \multimap \{\widetilde{L} \otimes
eq_{XL} \multimap \{\widetilde{R} \otimes eq_{XR}) \}\}\big)
\] 
from the sites $m$ in $L$ and $\eta$.\footnote{In the interest of
  clarity, we omit all references to the bigraph identifiers from Celf
  type constructors.}  Recall the three auxiliary operations
\emph{delete}, \emph{move}, and \emph{copy} that are triggered
depending on the cardinality $c = |\{x \mid x \in m, (x,y) \in
\eta\}|$.  In the case that is $ c=0$, we first colour all the direct
children of node $prnt_L\;x$ (using \ttt{tmp}) that contain the site
$x$ with colour \ttt{tmp\_prnt}.  Second we remove the colour
information from all of siblings of $x$ that are also present in $L$.
Third, for each coloured node, we start a recursive descent phase
(using \ttt{del}) to trigger the deletion of the node and its
children.
\begin{align*}
  \label{fig:impl_eta:del}\llb x \in m&, (x,y) \not\in \eta \rrb = \\
   & \ttt{has\_child\_p }(prnt_L \; x) \; (\ttt s^k\; N) \otimes
  \ttt{tmp } (prnt_L \; x) \; (s^k \; N) \\
   & \quad \otimes (\ttt{tmp } (prnt_L \;
  x) \; \ttt z \multimap \{ \botimes{k}{i=0} \ttt{tmp\_prnt } S_i
  \; D_i \\
   & \qquad \multimap \{ \botimes{k}{i=0} \ttt{prnt } S_i \; D_i 
  \otimes \ttt{del }(prnt_L \; x) \;N \otimes (\ttt{del } (prnt_L \; x) \; \ttt z \\
   & \qquad \quad \multimap \{ \ttt{has\_child\_p } (prnt_L \; x) \; (\ttt
  s^k \ttt z) \otimes \llb m\setminus\{x\}, \eta \rrb\})\}\})
\end{align*}
In the case that $c=1$, we do something very similar as in the
previous case, except that we move instead of delete. This case is
conceptually easier because we can skip the recursive descent phase.
\begin{align*}
     \llb x \in m&,  (x,y)  \in \eta\rrb 
   = \\
      & \ttt{has\_child\_p } (prnt_R \; y) \; (\ttt s^k \; N)
     \otimes  \ttt{tmp\_move } (prnt_L \; x) \; (s^k \; N) \\
      & \quad \otimes (\ttt{tmp\_move }(prnt_L \; x) \; \ttt z 
     \multimap \{\ttt{move }(prnt_L \; x) \; (prnt_R \; x) \; N \\
      & \qquad \otimes  (\ttt{move } (prnt_L \; x) \; (prnt_R \; x) \; \ttt z\\
      & \qquad \multimap \{ \ttt{has\_child\_p } (prnt_R \; y) \; 
     (\ttt s^k \; \ttt z) \otimes \llb m\setminus \{x\}, \eta\setminus
     \{(x,y) \}\rrb\})\})
\end{align*}
The case that $c>1$ is again similar, except that this time we need to
recursively copy the graph rooted in $prnt_L\;x$.  While copying we
are forced to create new nodes, ports, etc., which we get for free
from the \ttt{Exists} connective that is part of Celf.
\begin{align*}
 \llb x \in m&, (x,y) \in \eta \rrb = \\
      & \ttt{has\_child\_p } (prnt_R \;
       x) \; (\ttt s^k \; N) \otimes \ttt{tmp\_copy } (prnt_L \; x) (\ttt
       s^k \; N) \\
        & \quad \otimes (\ttt{tmp\_copy } (prnt_L \; x) \; \ttt z \multimap \{ 
       \botimes{k}{i=0} \ttt{tmp\_prnt } S_i \; D_i\\
        & \qquad \multimap \{ \botimes{k}{i=0} \ttt{prnt } S_i \; D_i \otimes
       \ttt{copy }(prnt_L \; x) \; (prnt_R \; y) \; N \\
        & \qquad \quad \otimes (\ttt{copy }
       (prnt_L \; x) \; (prnt_R \; y) \; \ttt z \\
        & \qquad \qquad \multimap
       \{\ttt{has\_child\_p } (prnt_R \; x) \; (\ttt s^k \; N) 
       \otimes \llb m,\eta\setminus \{(x,y) \}\rrb\})\}\})\\
\end{align*}
In the base case, we define $\llb \emptyset, \emptyset \rrb$ as a Celf
type constructor. Note that some cases, in particular move and copy,
need additional rules when the source node and the destination node
are the same.  This definition is well-formed because during the
recursive calls either, $m$, $\eta$, or both get smaller.



Finally, we address the question of adequacy.  Let $(B, \mathcal{R})$
be a \brs, and $\Gamma$ the intuitionistic Celf context that contains
the names of all ports, sites, roots, inner names, outer names, edges,
nodes, controls, the graph of the arity function in $B$, and the
translation of all reaction rules declared in $\mathcal{R}$. 

\begin{theorem}[Adequacy]  
  The agent $B$ reduces to agent $B'$ using the rules in $\mathcal{R}$
  if and only if (in {\upshape Celf})
  $\Pi{\Gamma}.\botimes{}{} \llb B' \rrb \multimap \{C\}$
  implies that $\Pi{\Gamma}.\botimes{}{} \llb B \rrb \multimap \{C\}$.
\end{theorem}
  
\begin{proof} 
  By induction on the reduction sequence, using the definition of
  $\llb(m,\eta)\rrb$ and
  Theorems~\ref{th:brs_encoding:ground_soundness},
  ~\ref{th:brs_encoding:ground_completeness}.
\end{proof}

\section*{Conclusion}
\label{section:conclusion}

In this paper we have described a model for bigraph reactive systems,
which we refer to as the \emph{bigraph relational model}. We have
shown that this model that is based on a multi-set rewriting system is
amenable to implementation.  The rewriting system ensures the validity
of the encoding with respect to the bigraph structural properties, and
we have shown that the semantics of \textsc{brs} is precisely captured
by the multi-set rewriting rules. Finally, we give an implementation
of the bigraph relational model in Celf, which is powerful tool:
Linearity allows us to implement the reaction rules directly, its
higher-order features take care of dynamic introduction of new rewrite
rules and the creation of fresh names while copying the place graph
where warranted.

\bibliographystyle{eptcs}
\bibliography{maxime}

\end{document}